\documentclass[11pt]{article}
\usepackage[T1]{fontenc}
\usepackage[utf8]{inputenc}
\usepackage{lmodern}
\usepackage[margin=1in]{geometry}
\usepackage{amsmath,amssymb,amsthm,mathtools}
\usepackage{microtype}
\usepackage{booktabs,array,tabularx}
\usepackage{enumitem}
\usepackage{aliascnt,etoolbox,needspace}
\usepackage[dvipsnames]{xcolor}
\usepackage[numbers,sort&compress]{natbib}
\usepackage{xurl}
\definecolor{LinkBlue}{RGB}{28,66,132}
\definecolor{CiteTeal}{RGB}{0,94,92}
\usepackage[colorlinks=true,linkcolor=LinkBlue,citecolor=CiteTeal,urlcolor=LinkBlue,
  bookmarksnumbered=true,pdfencoding=auto,psdextra]{hyperref}
\usepackage[nameinlink,capitalise,noabbrev]{cleveref}
\usepackage{bookmark}
\hypersetup{pdftitle={Explicit tensors beyond the linear flattening barrier},
  pdfauthor={Benjamin Lovitz}}
\providecommand{\doi}[1]{\href{https://doi.org/#1}{\nolinkurl{doi:#1}}}
\providecommand{\arxiv}[1]{\href{https://arxiv.org/abs/#1}{\nolinkurl{arXiv:#1}}}
\setlist[enumerate]{label=\textup{(\roman*)},leftmargin=2.1em,itemsep=2pt,topsep=4pt}
\allowdisplaybreaks[1]
\numberwithin{equation}{section}
\newtheorem{theorem}{Theorem}[section]
\newaliascnt{lemma}{theorem}
\newtheorem{lemma}[lemma]{Lemma}\aliascntresetthe{lemma}
\newaliascnt{proposition}{theorem}
\newtheorem{proposition}[proposition]{Proposition}\aliascntresetthe{proposition}
\newaliascnt{corollary}{theorem}
\newtheorem{corollary}[corollary]{Corollary}\aliascntresetthe{corollary}
\theoremstyle{definition}
\newaliascnt{definition}{theorem}
\newtheorem{definition}[definition]{Definition}\aliascntresetthe{definition}
\newaliascnt{example}{theorem}
\aliascntresetthe{example}
\theoremstyle{remark}
\newaliascnt{remark}{theorem}
\aliascntresetthe{remark}
\crefname{theorem}{Theorem}{Theorems}\Crefname{theorem}{Theorem}{Theorems}
\crefname{lemma}{Lemma}{Lemmas}\Crefname{lemma}{Lemma}{Lemmas}
\crefname{proposition}{Proposition}{Propositions}\Crefname{proposition}{Proposition}{Propositions}
\crefname{corollary}{Corollary}{Corollaries}\Crefname{corollary}{Corollary}{Corollaries}
\crefname{definition}{Definition}{Definitions}\Crefname{definition}{Definition}{Definitions}
\crefname{example}{Example}{Examples}\Crefname{example}{Example}{Examples}
\crefname{remark}{Remark}{Remarks}\Crefname{remark}{Remark}{Remarks}
\crefname{equation}{Equation}{Equations}\Crefname{equation}{Equation}{Equations}
\crefname{section}{Section}{Sections}\Crefname{section}{Section}{Sections}
\crefname{table}{Table}{Tables}\Crefname{table}{Table}{Tables}
\BeforeBeginEnvironment{theorem}{\Needspace{7\baselineskip}}
\BeforeBeginEnvironment{proposition}{\Needspace{6\baselineskip}}
\BeforeBeginEnvironment{lemma}{\Needspace{5\baselineskip}}
\BeforeBeginEnvironment{corollary}{\Needspace{5\baselineskip}}
\BeforeBeginEnvironment{definition}{\Needspace{5\baselineskip}}
\BeforeBeginEnvironment{remark}{\Needspace{4\baselineskip}}
\newcommand{\CC}{\mathbb C}

\newcommand{\BR}{\underline{\mathrm R}}
\newcommand{\TR}{\mathrm R}
\DeclareMathOperator{\rank}{rank}
\DeclareMathOperator{\tr}{tr}
\DeclareMathOperator{\Span}{span}
\DeclareMathOperator{\Gr}{Gr}
\DeclareMathOperator{\GL}{GL}

\DeclareMathOperator{\Mat}{Mat}
\DeclareMathOperator{\Pot}{Pot}
\DeclareMathOperator{\id}{id}

\newcommand{\Poly}[1]{\CC[t]_{<#1}}
\newcommand{\Kos}[2]{\mathcal K_{#1}(#2)}

\title{Explicit tensors beyond the linear flattening barrier}
\author{Benjamin Lovitz\thanks{Concordia University, Montreal, Canada. benjamin.lovitz@concordia.ca}}
\date{\today\\[2pt]}

\begin{document}
\maketitle
\begin{abstract}
We construct explicit $n \times n \times n$ tensors of border rank at least $3n-o(n)$, improving the previous record of $(2+\varepsilon)n$ due to Landsberg--Micha{\l}ek~\cite{LM2025}. We also prove that the linear flattening method cannot be used to establish lower bounds on border rank beyond $2n-1$. When $n$ is odd, we prove that this bound is achieved by Koszul flattenings. When $n$ is even, a result of Landsberg~\cite{Landsberg2015} shows that Koszul flattenings can achieve $2n-2$, leaving an open gap of size one. Our $2n-1$ bound improves the best known $6n-4$ linear flattening barrier due to Garg et al.\ and Buczy{\'n}ski~\cite{GMOW2019,Buczynski2026}. Combined, these results show that our $3n-o(n)$ construction, as well as Landsberg--Micha{\l}ek's construction, provide explicit examples of tensors with higher border rank than any linear flattening can achieve.
\end{abstract}

\section{Introduction}\label{sec:introduction}

Let $(\CC^n)^{\otimes 3}$ be the space of $n \times n \times n$ tensors. We say that a tensor has \textit{rank one} if it can be written as $x \otimes y \otimes z:=(x_i y_j z_k)_{i,j,k}$ for some nonzero vectors $x,y,z \in \CC^n$. A \textit{decomposition} of a tensor $T\in (\CC^n)^{\otimes 3}$ is an expression of $T$ as a sum of rank-one tensors, and the \textit{tensor rank} $\TR(T)$ is the smallest number of rank-one tensors needed in a decomposition of $T$. The \textit{border rank} $\BR(T)$ is the smallest number $r$ for which $T$ can be written as a limit of tensors, each of tensor rank at most $r$. We let $\sigma_r \subseteq (\CC^n)^{\otimes 3}$ be the subset of tensors of border rank at most $r$.

The tensor rank of a tensor has a direct algorithmic interpretation: a decomposition into $r$ terms gives an algorithm for evaluating the associated bilinear map using $r$ products of linear forms applied to each input. Border rank captures the corresponding approximate notion, allowing the bilinear map to arise as a limit of maps of tensor rank at most $r$. A particularly well-motivated example is matrix multiplication, where decompositions of the associated tensor lead to algorithms for multiplying matrices~\cite{BCS1997,Landsberg2012}. A common yet difficult goal in complexity theory is to exhibit explicit objects with high complexity. Well-known examples include Boolean functions requiring large Boolean circuits~\cite{RazborovRudich1997} and polynomials requiring large arithmetic circuits~\cite{Valiant1979}. In this work, we address this goal for border rank by constructing explicit tensors with provably large border rank.

It is known that the maximum border rank of an $n\times n \times n$ tensor is $\lceil n^3/(3n-2)\rceil$ for all $n \neq 3$, and $5$ for $n=3$~\cite{Strassen1983,Lickteig1985}.  Despite this, the current best known lower bound on the border rank of an explicit tensor is $(2+\varepsilon)n$ for any $0 < \varepsilon <1/42$ for sufficiently large $n$, due to Landsberg--Micha{\l}ek~\cite[Theorem~1.2]{LM2025}. Prior to this, Landsberg
\cite{Landsberg2015} and Derksen--Makam \cite{DM2018,DM2019} found tensors of border rank nearly $2n$ using Koszul flattenings. For the matrix multiplication tensor, Landsberg--Micha{\l}ek
\cite{LM2018} proved a lower bound $2q^2-\lceil \log_2q \rceil-1$ in ambient factor dimension
$n=q^2$. For the standard tensor rank, Alexeev--Forbes--Tsimerman \cite{AFT2011} constructed explicit $0/1$ tensors
of tensor rank $3n-\Theta(\log n)$ over arbitrary fields, and Landsberg~\cite{Landsberg2014} proved that the matrix multiplication tensor has tensor rank at least $3q^2-o(q^2)$ over the complex numbers.

Our first result is to improve the best known lower bound on the border rank of an explicit tensor to $3n-o(n)$.

\begin{theorem}\label{thm:explicit}
There is a polynomial-time explicit tensor
$T_n\in(\CC^n)^{\otimes3}$ such that
\begin{equation}\label{eq:explicit-main}
 \BR(T_n)\ge3n- 11 n^{3/4}.
\end{equation}
\end{theorem}
Here, we say that a tensor is \textit{polynomial-time explicit} if it
has integer entries and a deterministic algorithm records all of its
entries in binary in time polynomial in $n$, following
\citet[Definition~1.1]{LM2025}.

Can we do better? A common method of lower bounding the border rank is to find a nonzero polynomial $f(T)$ on $(\CC^n)^{\otimes 3}$ that is zero on $\sigma_r$. Such polynomials are called \textit{nontrivial equations} for $\sigma_r$, and they certify $\BR(T) > r$ for every tensor $T$ that is nonzero on $f$. An attractive way to find equations for $\sigma_r$ is via \textit{linear flattenings}. A linear flattening is a nonzero linear map from tensors to matrices $L: (\CC^n)^{\otimes 3} \rightarrow \Mat_{u,v}(\CC)$, where $u,v$ are positive integers. Let
\begin{equation}\label{eq:potency-intro}
 \rho(L)=\max_{x,y,z\in\CC^n}\rank L(x\otimes y\otimes z),\qquad
 \Pot(L)=\frac{\max_T\rank L(T)}{\rho(L)}.
\end{equation}
Following~\cite[Definition 1.7]{GMOW2019}, $\Pot(L)$ is called the
\emph{potency} of $L$. The \emph{optimal potency} $\Pi(n)$ is the supremum of $\Pot(L)$ over all nonzero linear flattenings, with no restriction on the output dimensions $u,v$. If $\Pot(L) > r$ then some $(\rho(L) r+1) \times (\rho(L) r+1)$ minor of $L(T)$ gives a nontrivial equation for $\sigma_r$. We call this the \textit{linear flattening method} for obtaining nontrivial equations for $\sigma_r$. The largest integer $r$ for which the linear flattening method can produce a nontrivial equation for $\sigma_r$ is $\lceil \Pi(n) \rceil -1$.\footnote{We note that there could be a minor of some $L(T)$ that happens to furnish a nontrivial equation for some larger $\sigma_r$ via a different argument. We do not include such cases in our definition of the linear flattening method.}

The first linear upper bound $\Pi(n)\leq 8n$ was obtained by Efremenko et al.~\cite{EGOW2018}, which was subsequently improved to $\Pi(n) \leq 6n-4$ by Garg et al.~\cite[Corollary~8.5]{GMOW2019}. Buczy{\'n}ski's recent account~\cite{Buczynski2026} explains this number through cactus varieties. See also Kothari--Moitra--Wein \cite[Theorems~2.9 and~2.10]{KMW2025}.

In terms of lower bounds, it is known that $\Pi(n) \geq 2n-2$ when $n$ is even~\cite[Theorem 1.2]{Landsberg2015}, and $\Pi(n) > 2n-4+o(1)$ when $n$ is odd~\cite[Theorem 6.1]{DM2018}, and both results are obtained by Koszul flattenings.

We determine $\Pi(n)$ exactly for odd $n$, and to within a vanishing additive gap for even $n$.

\begin{theorem}\label{thm:barrier}
For every $n\ge2$,
\begin{equation}\label{eq:sharpness-main}
 2n-2\le\Pi(n)\le 2n-2 + \frac{2}{n+1}.
\end{equation}
For odd $n$, the upper bound in \Cref{eq:sharpness-main} holds with
equality and is attained by the central Koszul flattening.
\end{theorem}
The central Koszul flattening $\mathcal K_{(n-1)/2}$ is defined in \Cref{def:koszul}. In particular, we obtain the following:
\begin{corollary}\label{cor:sharpness}
For every $n\ge2$, the linear flattening method
can yield a nontrivial equation for $\sigma_r$ only if
$r\le2n-2$. For odd $n\ge3$, equality is attained by the
central Koszul flattening. For even $n\ge2$, Koszul flattenings
attain $r=2n-3$.
\end{corollary}
This corollary follows directly from~\Cref{thm:barrier} and the discussion above, with the exception of the even-$n$ construction, which is given in~\cite[Theorem 1.2]{Landsberg2015} via the composition of a Koszul flattening with the projection onto $\CC^{n-1}$ in the first factor. For odd $n$, this improves (optimally) the bound $r=2n-4$ obtained in~\cite[Theorem 6.1]{DM2018}, generalizes the $n=5$ case observed in~\cite[Proposition~1.3]{Landsberg2015} to arbitrary $n$, and confirms the general behaviour suggested in~\cite[Remark~6.2]{DM2018}.

\subsection{Proof overviews}

\subsubsection*{Proof overview for~\Cref{thm:explicit}}\label{sec:overview}

We construct $T_n$ as follows. For positive integers $e,m$, let
$T_{e,m}\in(\CC^{me^2})^{\otimes3}$ be the tensor with entries
\[
 (T_{e,m})_{(i,\alpha),(j,u),(k,v)}
 =
 \begin{cases}
  2^{\alpha(j+u+v)},&k=i+j,\\
  0,&k\ne i+j,
 \end{cases}
\]
where
\[
 0\le i<m,\qquad 0\le j,k<me,\qquad
 0\le\alpha<e^2,\qquad 0\le u,v<e.
\]
We then take $e=\lfloor n^{1/4}\rfloor$ and $m=\lfloor n/e^2\rfloor$
and obtain $T_n$ by padding $T_{e,m}$ with zeroes.
This construction is polynomial-time explicit because each nonzero entry
is of the form $2^r$ for $r=n^{O(1)}$.

Our construction and proof is similar to that of
\citet{LM2025}. Their tensor coordinates are indexed by $i,j,k$, with
nonzero entries supported on $k=i+j$, while our tensor has the same support
pattern at the level of blocks.
Our overall proof technique is also similar: we use border substitution
to reduce the task of lower bounding $\BR(T_{e,m})$ to that of uniformly
lower bounding the border rank of its restrictions to lower-dimensional subspaces of the first factor,
for which we use Koszul flattenings. The $3n-o(n)$ lower bound follows from roughly $n$ from border substitution and $2n$ from Koszul flattenings of the restrictions. To lower bound the ranks of the Koszul flattenings, we analyze the
block structure of the restricted tensor. In short, the symmetry $k=i+j$ together with the low-dimensional restriction allows us to use a pigeonhole argument to find a block upper
triangular submatrix of the Koszul flattening of a further restriction of $T_{e,m}$, whose rank we can analyze.

\subsubsection*{Proof overview for~\Cref{thm:barrier}}
\emph{Lower bound:} The $2n-2$ lower bound for even $n=2p+2\ge4$ follows from
Landsberg~\cite[Theorem~1.2]{Landsberg2015}, who constructs an $n \times n \times n$ tensor for which the flattening $L(T)$, obtained from projecting the first factor onto $\CC^{n-1}$ and applying the $p=n/2-1$ Koszul flattening, has rank $n\binom{2p+1}p$ and satisfies $\rho(L)=\binom{2p}p$, hence $\Pot(L)=n\binom{2p+1}{p}/\binom{2p}{p}=2n-2$. For $n=2$, an ordinary flattening gives the same bound.

\emph{Upper bound:}
We prove the upper bound $\Pi(n) \leq 2n-2 + \frac{2}{n+1}$ via a probabilistic argument. Let $L$ be a linear flattening with $\rho(L)=\rho$ and $\Pot(L)=r/\rho$. Then we can assume that $L : (\CC^n)^{\otimes 3} \rightarrow \Mat_r(\CC)$, $L(T)$ has full rank for some $T$, and $L(x \otimes y \otimes z)$ has rank at most $\rho$ for all $x,y,z \in \CC^n$. Using the Kempf--Ness theorem, we can assume that $L(T)=\sum_{ijk} T_{ijk} M_{ijk} $ for some matrices $M_{ijk} \in \Mat_r(\CC)$ for which
\begin{align}\label{eq:HK}
H:=\sum_{ijk} M_{ijk} M_{ijk}^{\dagger}, \qquad K:= \sum_{ijk} M_{ijk}^{\dagger} M_{ijk}
\end{align}
satisfy $H=K=I_r$. Let $M=M(x,y,z)=L(x\otimes y \otimes z)$. Taking the average over the uniform spherical measures on each copy of $\CC^n$, we prove that
\begin{align}\label{eq:moment_bounds}
 \mathbb E\|M\|_2^2=\frac{\tr H}{n^3}=\frac{r}{n^3},
 \qquad
 \mathbb E \|M\|_4^4
 \le\frac{\|H\|_2^2+\|K\|_2^2}{n^4(n+1)}=\frac{2r}{n^4(n+1)},
\end{align}
where $\| \cdot\|_p$ is the Schatten $p$-norm (e.g. $\|M\|_2$ is the Hilbert--Schmidt norm). While the first equation is well-known, we were unable to find the upper bound on $\mathbb E \|M\|_4^4$ in the literature, and we suggest that it could be of independent interest. From here, two applications of Cauchy-Schwarz gives
\[
\frac{r^2}{n^6}=\left(\mathbb E\|M\|_2^2\right)^2
\le
\mathbb E\|M\|_2^4
\le
\rho\, \mathbb E \|M\|_4^4\leq \frac{2 \rho r}{n^4(n+1)},
\]
which implies the desired upper bound $\Pot(L)=r/\rho\leq 2n-2 + \frac{2}{n+1}$.

\emph{Sharpness for odd $n$:} To prove equality for odd $n=2p+1$, we construct an $n \times n \times n$ tensor $S_n$ for which the central Koszul flattening $\Kos{p}{S_n}$ has full rank. This suffices because $\Kos{p}{S_n}$ is an $n\binom np\times n\binom np$ matrix and $\rho(\mathcal K_p)=\binom{n-1}p$ (see \Cref{prop:koszul-bound}), so full rank gives potency $n\binom np/\binom{n-1}p=2n-2+2/(n+1)$. Let $A=B=C=\CC^n$. The tensor $S_n$ is given by
\begin{equation*}
 S_n=\sum_{i=1}^n a_i\otimes X_i,
 \qquad
 X_i(f)=f(i)v_{p+1}(i)\oplus f'(i)v_p(i),
\end{equation*}
where $v_q(t):=(1,t,\ldots,t^{q-1})^{\mathsf T}$, $X_i : B^* \rightarrow C$, and we identify $B^*\cong \Poly{n}$ and $C\cong \CC^{p+1}\oplus\CC^p$. A vector in the kernel of $\Kos{p}{S_n}$ is a linear combination of terms of the form $a_I \otimes f_I$, where $a_I:=a_{i_1}\wedge \cdots \wedge a_{i_p} \in \Lambda^p(A)$ is the standard basis indexed by $p$-subsets $I =\{i_1<\cdots<i_p\}\subseteq [n]$, and $f_I \in B^*= \Poly{n}$. We then prove that the polynomials $f_I$ have at least $n$ roots counted with multiplicity, so they are identically zero. To do this, we use the kernel condition and the Vandermonde structure of the vectors $v_q(t)$ to show that both $f_I(i)=0$ and $f'_I(i)=0$ for all $i\notin I$, giving $2(p+1)=n+1$ roots with multiplicity.

\subsection{Open questions}

In this section we list several questions left open by this work.

\begin{enumerate}
\item Can we find explicit tensors of border rank (or even, tensor rank) greater than $3n$? This will require techniques beyond the method of border substitution + linear flattenings on the restricted tensors, such as border apolarity~\cite{BuczynskaBuczynski2021,ConnerHarperLandsberg2023} or the recently introduced Kronecker-Young flattenings~\cite{DolezalekMichalek2026}.
\item For an even integer $n$, what is the largest $r$ for which the linear flattening method can yield nontrivial equations for $\sigma_r$? By~\Cref{cor:sharpness}, $r \in \{2n-3, 2n-2\}$. For $n=2$, the answer is $r=1$ because $\sigma_2=(\CC^2)^{\otimes 3}$. For odd $n$, our results show that $r=2n-2$, and this is achieved by the central Koszul flattening.
\item To what extent can these techniques be generalized to exhibit explicit non-square tensors beyond the linear flattening barrier?
\end{enumerate}
\subsection{Organization}

Relevant background is covered in \Cref{sec:background}. The explicit tensors of high border rank presented in \Cref{thm:explicit} are constructed in~\Cref{sec:construction}. The linear flattening barrier presented in \Cref{thm:barrier} is proven in~\Cref{sec:barrier-tensor}, and sharpness for odd $n$ is proven in~\Cref{sec:sharpness}.

\paragraph{Use of artificial intelligence} The author used GPT-6 Astra interactively throughout every stage of preparing this manuscript, including to obtain proofs of the main results, which the author subsequently verified and simplified. The author takes full responsibility for the correctness, exposition, and attribution in this manuscript.

\paragraph{Acknowledgments} The author thanks Harm Derksen for helpful feedback on a draft of this manuscript. The author acknowledges support from NSERC Discovery Grant number RGPIN-2026-05413.

\section{Background}\label{sec:background}

We now introduce some necessary background for this work. All vector spaces are finite-dimensional over $\CC$.

Let $A, B, C$ be vector spaces.  We will frequently regard a tensor
$T\in A\otimes B\otimes C$ as a linear map
$A^*\to B\otimes C$, and we regard $B\otimes C$
as the space of linear maps $B^*\to C$.  For $W\subseteq A^*$,
write $T|_W\in W^*\otimes B\otimes C$ for the restricted contraction
map.  The tensor is \emph{$A$-concise} if its contraction map on
$A^*$ is injective.

For a complex matrix $X$, write $X^\dagger$ for its Hermitian adjoint.
If $\sigma_j(X)$ are its singular values, its \emph{Schatten norms} are
\[
 \|X\|_p=\left(\sum_j\sigma_j(X)^p\right)^{1/p}
 \quad(1\le p<\infty),
 \qquad
 \|X\|_\infty=\max_j\sigma_j(X).
\]
In particular, $\|X\|_\infty$ is the Euclidean operator norm, and
$\|X\|_2$ is the Hilbert--Schmidt norm.

\begin{definition}\label{def:vandermonde}
For an integer $d\ge1$ and $x_1,\ldots,x_d\in\CC$, the
\emph{Vandermonde matrix} is
\[
 V(x_1,\ldots,x_d)
 =\bigl[x_j^{i-1}\bigr]_{1\le i,j\le d}
 =\begin{pmatrix}
  1&1&\cdots&1\\
  x_1&x_2&\cdots&x_d\\
  \vdots&\vdots&&\vdots\\
  x_1^{d-1}&x_2^{d-1}&\cdots&x_d^{d-1}
 \end{pmatrix}.
\]
Its determinant is \cite[Equation~(1.3.13)]{DLMF}
\begin{equation}\label{eq:vandermonde-determinant}
 \det V(x_1,\ldots,x_d)
 =\prod_{1\le i<j\le d}(x_j-x_i).
\end{equation}
Thus $V(x_1,\ldots,x_d)$ is invertible if and only if
$x_1,\ldots,x_d$ are pairwise distinct.
\end{definition}

For a positive integer $h$, we let $v_h(t):=(1,t,\dots, t^{h-1})^{\mathsf T}$, we let $\Poly{h}$ be the space of univariate polynomials of degree $<h$, and we let $[h]:=\{1,\dots, h\}$.

For an ordered basis $z_1,\ldots,z_n$ of $A$ and a $p$-subset $I=\{i_1<\cdots<i_p\}\subseteq [n]$, we let
$z_I=z_{i_1}\wedge\cdots\wedge z_{i_p}\in\Lambda^pA$, with
$z_\varnothing=1$.

We now define Koszul flattenings. See~\cite{LandsbergOttaviani2013} and \cite[Equation 2.1]{LO2015}. See also \cite[Section~2.2.2]{KMW2025} for historical notes on Koszul flattenings and the more general Young flattenings.
\begin{definition}[Koszul flattenings]\label{def:koszul}
For each $0\le p<\dim A$, define the \textit{Koszul flattening} on $A \otimes B \otimes C$ by
\begin{align}\label{eq:koszul-definition}
 \Kos{p}{a \otimes b \otimes c}:\Lambda^pA \otimes B^*&\longrightarrow
              \Lambda^{p+1} A \otimes C,\\
 \omega\otimes\varphi&\longmapsto
       (a \wedge\omega)\otimes \varphi(b) c
 \notag
\end{align}
for $a \in A, \; b \in B, \; c\in C$, extended linearly.
\end{definition}

The following proposition is well-known. We prove it for completeness.

\begin{proposition}[Theorem~2.1 in \cite{LO2015}]\label{prop:koszul-bound}
Let $A,B,C$ be nonzero complex vector spaces of dimensions
$n_A,n_B,n_C$, and let $0\le p<n_A$.
The degree-$p$ Koszul flattening satisfies
\begin{equation*}
 \rho(\mathcal K_p)=\binom{n_A-1}{p},
\end{equation*}
and for every $S\in A\otimes B\otimes C$,
\begin{equation*}
 \BR(S)\ge
 \frac{\rank\Kos{p}{S}}{\binom{n_A-1}{p}}.
\end{equation*}
\end{proposition}

\begin{proof}
For a nonzero rank-one tensor $S=a\otimes b\otimes c$,
the Koszul map is
\[
 \Kos{p}{S}(\omega\otimes\varphi)
 =
 (a\wedge\omega)\otimes\varphi(b)c,
\]
i.e. it is the tensor product of
\[
 W_a:\Lambda^pA\longrightarrow\Lambda^{p+1}A,
 \qquad \omega\longmapsto a\wedge\omega
\]
with the rank-one map
\[
 B^*\longrightarrow C,
 \qquad
 \varphi\longmapsto\varphi(b)c.
\]
Extend $a=a_1$ to a basis $a_1,\ldots,a_{n_A}$ of $A$.
The map $W_a$ annihilates the exterior basis vectors containing
$a_1$. Its images on the remaining basis vectors are
\[
 a_1\wedge a_{i_1}\wedge\cdots\wedge a_{i_p},
 \qquad
 2\le i_1<\cdots<i_p\le n_A,
\]
which are linearly independent.
There are $\binom{n_A-1}{p}$ such vectors, so
\[
 \rank\Kos{p}{a\otimes b\otimes c}
 =
 \binom{n_A-1}{p}.
\]
This proves the first equation.

If $S$ is a sum of $r$ rank-one tensors, linearity of
$\mathcal K_p$ and subadditivity of matrix rank give
\[
 \rank\Kos{p}{S}\le r\binom{n_A-1}{p}.
\]
Since this is a closed condition, the same inequality holds when $\BR(S)\le r$.
\end{proof}


We now recall the border-substitution method of
\cite{LM2017}; see also \cite[Equation~(2.1)]{LM2025}. We include a proof for completeness. Let $\Gr(q,A^*)$ be the Grassmannian of $q$-dimensional subspaces of $A^*$.

\begin{lemma}[Equation~(2.1) in \cite{LM2025}]\label{lem:substitution}
Let $T\in A\otimes B\otimes C$ be $A$-concise, where
$\dim A=\dim B=\dim C=n$. For $1\le q\le n$,
\begin{equation}\label{eq:substitution}
 \BR(T)\ge n-q+
       \min_{W\in\Gr(q,A^*)}\BR(T|_W).
\end{equation}
\end{lemma}

\begin{proof}
Put $r=\BR(T)$ and choose $T_j\to T$ with $\TR(T_j)\le r$.
For all sufficiently large $j$, the tensors $T_j$ are
$A$-concise.  The first-factor vectors in a rank decomposition
of $T_j$ therefore span $A$.  Choose $n-q$ independent vectors
among them, and let $W_j\subseteq A^*$ be their common
annihilator.  Then $W_j$ is a $q$-dimensional subspace, and
\[
 \TR(T_j|_{W_j})\le r-n+q.
\]
By compactness of the complex Grassmannian, we can pass to a convergent subequence
with $W_j\to W$.  The restricted tensors $T_j|_{W_j}$
then converge to $T|_W$, so $\BR(T|_W)\le r-n+q$.
\end{proof}

\section{An explicit construction}\label{sec:construction}
In this section we prove~\Cref{thm:explicit}. Throughout this section, fix an integer $e\ge2$ and put $s=2e-1$.
The $(e-1)$-th Koszul matrix of a tensor in
$\CC^s\otimes\CC^e\otimes\CC^e$ is square of size
$e\binom{s}{e-1}$, since $\binom{s}{e}=\binom{s}{e-1}$.

We first establish two lemmas. The first is an elementary sufficient condition for the central Koszul matrix to be invertible.

\begin{lemma}\label{lem:general-position}
Let $z_1,\ldots,z_s$ be a basis of $E$, and let
$b_1,\ldots,b_s\in B$ and $c_1,\ldots,c_s\in C$,
where $\dim B=\dim C=e$.
Let
\[
 T=\sum_{\ell=1}^s z_\ell\otimes b_\ell\otimes c_\ell.
\]
If every subset of $e$ vectors in each collection is linearly independent, then
$\Kos{e-1}{T}$ is invertible.
\end{lemma}

\begin{proof}
Let $w$ be a vector in the kernel of $\Kos{e-1}{T}$, which we can write as
\[
 w=\sum_{|I|=e-1}z_I\otimes\varphi_I,\qquad\varphi_I\in B^*.
\]
For an $e$-subset $J \subseteq [s]$, write $J=\{j_1 < \cdots < j_e\}$. By \Cref{eq:koszul-definition}, it holds that
\begin{align*}
 0=\Kos{e-1}{T}w
 &=\sum_{|I|=e-1}\sum_{\ell\notin I}
   (z_\ell\wedge z_I)\otimes\varphi_I(b_\ell)c_\ell\\
 &=\sum_{|J|=e}z_J\otimes
   \left(\sum_{r=1}^e(-1)^{r-1}
   \varphi_{J\setminus\{j_r\}}(b_{j_r})c_{j_r}\right),
\end{align*}
so, for every such $J$,
\[
 \sum_{r=1}^e(-1)^{r-1}
 \varphi_{J\setminus\{j_r\}}(b_{j_r})c_{j_r}=0.
\]
Independence of the $c_{j_r}$ gives
$\varphi_I(b_j)=0$ for every $(e-1)$-subset $I\subseteq[s]$ and $j\notin I$.
For each $I$, these are $e$ independent vectors $b_j$, so
$\varphi_I=0$. Thus $\Kos{e-1}{T}$ is injective, and hence invertible.
\end{proof}

\begin{definition}\label{def:local-evaluation}
Let $f_1,\dots, f_s \in \Poly{e^2}$ be monic polynomials of distinct degrees.
Using the standard bases of
$\CC^s\otimes\CC^e\otimes\CC^e$, define
\[
 Q_f(t):=
 \sum_{\ell=1}^s\sum_{u,v=0}^{e-1}
 f_\ell(2^{u+v}t)\,z_\ell\otimes b_u\otimes c_v,
 \qquad
 M_f(t):=\Kos{e-1}{Q_f(t)}.
\]
\end{definition}

For our explicit construction $T_{e,m} \in A \otimes B \otimes C$ below, we will establish that for any restriction $T_{e,m}|_{W}$ of appropriate dimension, the resulting Koszul matrix contains a block upper triangular submatrix with diagonal blocks given by the matrices $M_f(t)$. The following lemma gives us a method to lower bound the rank of this block matrix, and hence the border rank of $T_{e,m}|_{W}$, allowing us to apply the border substitution method to $T_{e,m}$. We note that the degree counting argument used in the proof of the following lemma is similar to the one used in the proof of~\cite[Theorem~14]{GK2016}.

\begin{lemma}\label{lem:local-nonzero}
For $f$ as in \Cref{def:local-evaluation} and any distinct points
$t_0,\ldots,t_{L-1}\in\CC$,
\begin{equation}\label{eq:local-kernel-dimension}
 \sum_{j=0}^{L-1}\dim\ker M_f(t_j)\le e\binom{s}{e-1}(e^2-1).
\end{equation}
\end{lemma}

\begin{proof}
Put $d_\ell=\deg f_\ell$ and
$v_e(x)=(1,x,\ldots,x^{e-1})^{\mathsf T}$.
For $t\ne0$, let $G_t \in \GL(\CC^s)$ be defined by $G_tz_\ell=t^{-d_\ell}z_\ell$. Then
\[
 (G_t\otimes\id\otimes\id)Q_f(t)
 \longrightarrow
 \sum_{\ell=1}^s
 z_\ell\otimes v_e(2^{d_\ell})\otimes v_e(2^{d_\ell})
 \qquad(t\to+\infty).
\]
Every $e$-subset of the vectors $v_e(2^{d_\ell})$ is linearly independent by the
Vandermonde determinant. Since invertible matrices on the factors do not change the ranks of the Koszul flattenings, \Cref{lem:general-position} and continuity shows that
$g(t):=\det M_f(t)$ is nonzero for $t$ large enough, hence $g(t)$ is not identically zero.
By \Cref{eq:koszul-definition}, the entry of $M_f(t)$ in row
$z_J\otimes c_v$ and column $z_I\otimes b_u^*$ is zero unless
$J=I\cup\{\ell\}$ for some $\ell\notin I$; in that case it is
$\pm f_\ell(2^{u+v}t)$. Thus every nonzero entry of $M_f(t)$ has
degree $d_\ell\le e^2-1$. Each term in the determinant expansion is
a product of $e\binom{s}{e-1}$ matrix entries, so
$\deg g\le e\binom{s}{e-1}(e^2-1)$.

Fix $a\in\CC$ and put $r=\dim\ker M_f(a)$.
Recalling that $M_f(a)$ is a square matrix, Gaussian elimination on
the numerical matrix $M_f(a)$ produces $r$ zero
rows. These row operations are encoded by an invertible matrix $P$ over
$\CC$, which may depend on $a$ but is independent of $t$.
Apply the same row operations to the polynomial matrix $M_f(t)$.
Every entry in the corresponding $r$ rows of $PM_f(t)$ is a polynomial
that vanishes at $t=a$, and hence is divisible by $t-a$.
Factoring $t-a$ from each of these rows shows that $(t-a)^r$ divides
$\det(PM_f(t))=\det(P)g(t)$. Since $\det(P)$ is a nonzero constant,
$(t-a)^r$ also divides $g(t)$.
Applying this at the distinct points $t_j$ gives
\[
 \sum_{j=0}^{L-1}\dim\ker M_f(t_j)
 \le\deg g\le e\binom{s}{e-1}(e^2-1).
\]
This completes the proof.
\end{proof}

\Needspace{8\baselineskip}
\subsection{The tensor and its border rank}

For an integer $m\ge e^2$, let $A,B,C$ be complex vector spaces with bases
\[
 \begin{aligned}
 A &: \quad a_{i,\alpha}
       &&(0\le i<m,\ 0\le\alpha<e^2),\\
 B &: \quad b_{j,u}
       &&(0\le j<me,\ 0\le u<e),\\
 C &: \quad c_{k,v}
       &&(0\le k<me,\ 0\le v<e).
 \end{aligned}
\]
All three spaces have dimension $me^2$.
Define $T_{e,m}\in A\otimes B\otimes C$ by
\begin{equation}\label{eq:tensor-entry}
 (T_{e,m})_{(i,\alpha),(j,u),(k,v)}
 =
 \begin{cases}
  2^{\alpha(j+u+v)},&k=i+j,\\
  0,&k\ne i+j.
 \end{cases}
\end{equation}
Write $A_i=\Span\{a_{i,\alpha}:0\le\alpha<e^2\}$,
$B_j=\Span\{b_{j,u}:0\le u<e\}$, and
$C_k=\Span\{c_{k,v}:0\le v<e\}$.

As we will be applying the border substitution method to $T_{e,m}$ (see~\Cref{lem:substitution}), we first need to verify that it is $A$-concise. For fixed \(i\), the \(A_i\)-slices are supported only on the blocks \(B_j\otimes C_k\) with \(k-j=i\). Hence the supports corresponding to distinct values of \(i\) are disjoint, so it suffices to prove that the restriction $T_{e,m}|_{A_i^*}$ is injective for each $i$. For each $i$, taking $u=v=0$ and $0\le j<e^2$
gives the invertible Vandermonde matrix
$[2^{\alpha j}]_{0\le j,\alpha<e^2}$. Since the columns of this matrix are coordinates of the slices $T_{e,m}(a_{i,0}^*),\dots, T_{e,m}(a_{i,e^2-1}^*)$, these slices are linearly independent, so $T_{e,m}|_{A_i^*}$ is injective. These entries exist since $i+j\le m+e^2-2<me$ when $0\le j<e^2$.

The $A$-conciseness of $T_{e,m}$, together with the following lemma, will allow us to apply the border substitution method to $T_{e,m}$.

\begin{lemma}\label{lem:uniform-restrictions}
For every $ms$-dimensional subspace $W\subseteq A^*$,
\[
 \BR(T_{e,m}|_W)\ge s(me-m-e^2).
\]
\end{lemma}

\begin{proof}
Let $F_i=\bigoplus_{i'\le i}A_{i'}^*$, with the convention $F_{-1}=0$.
Since $F_{m-1}=A^*$ and $\dim W=ms$, we have
\[
 \sum_{i=0}^{m-1}
 \bigl(\dim(W\cap F_i)-\dim(W\cap F_{i-1})\bigr)=ms.
\]
By the pigeonhole principle, there exists $i$ such that
$\dim(W\cap F_i)-\dim(W\cap F_{i-1})\ge s$.
Fix such an $i$. The coordinate projection $\pi_i:F_i\to A_i^*$
restricts to a map on
$W\cap F_i$ with kernel $W\cap F_{i-1}$. By rank-nullity, its
image has dimension at least $s$. Thus we can choose
$\varphi_1,\ldots,\varphi_s\in W\cap F_i$ whose projections to $A_i^*$
are linearly independent.

For a nonzero vector $v\in A_i^*$, let
$\iota(v)=\max\{\alpha:v(a_{i,\alpha})\ne0\}\in\{0,1,\ldots,e^2-1\}$
be the index of its last nonzero entry.
Since $\pi_i(\varphi_1),\ldots,\pi_i(\varphi_s)$ are linearly
independent, there is an invertible matrix $P\in \Mat_s(\CC)$
such that, putting $\psi_\ell=\sum_{q=1}^sP_{\ell q}\varphi_q$,
the last nonzero entry of each $\pi_i(\psi_\ell)$ equals one and
\[
 \iota(\pi_i(\psi_1))<\cdots<\iota(\pi_i(\psi_s)).
\]
Replace $\varphi_\ell$ by $\psi_\ell$; these covectors still lie in
$W\cap F_i$ and project to linearly independent vectors in $A_i^*$.
Furthermore, under this choice, the polynomials
\[
 f_\ell(t):=
 \sum_{\alpha=0}^{e^2-1}\varphi_\ell(a_{i,\alpha})t^\alpha
 \qquad(1\le\ell\le s)
\]
are monic of distinct degrees.

Let $E^*=\Span\{\varphi_1,\ldots,\varphi_s\}$ and let
$z_1,\ldots,z_s\in E$ be the dual basis to
$\varphi_1,\ldots,\varphi_s$. Put $S=T_{e,m}|_{E^*}$.
Recall that $\varphi_\ell\in F_i=\bigoplus_{i'\le i}A_{i'}^*$.
For $0\le j<me-i$, \Cref{eq:tensor-entry} shows that contraction
by the $A_{i'}^*$ component of $\varphi_\ell$ maps $B_j^*$ into
$C_{j+i'}$. Indeed, every nonzero entry of $T_{e,m}$ with first
block index $i'$ and second block index $j$ has third block index
$k=j+i'$. Hence, the contraction
$S(\varphi_\ell):B^*\to C$ satisfies
\[
 S(\varphi_\ell)(B_j^*)\subseteq\bigoplus_{i'=0}^i C_{j+i'}.
\]

Let $K$ be the matrix obtained from $\Kos{e-1}{S}$ by retaining the column blocks
$\Lambda^{e-1}E\otimes B_j^*$ and row blocks
$\Lambda^eE\otimes C_{j+i}$ for $0\le j<me-i$, discarding all
other block columns and rows. Order both sets of retained blocks
by increasing $j$.
For $\omega\in\Lambda^{e-1}E$ and $0\le u<e$,
\Cref{eq:koszul-definition} gives
\[
 \Kos{e-1}{S}(\omega\otimes b_{j,u}^*)
 =\sum_{\ell=1}^s(z_\ell\wedge\omega)
       \otimes S(\varphi_\ell)(b_{j,u}^*).
\]
Thus a contribution from the $A_{i'}^*$ component of a covector
can reach the retained row block indexed by $j'$ only if
\[
 j'+i=j+i',\qquad\text{hence }j'\le j.
\]
In particular, every block of $K$ below the diagonal, where $j'>j$, is
zero. This proves that $K$ is block upper triangular.

By \Cref{eq:tensor-entry} and the fact that $\varphi_\ell\in F_i$, we have
\[
 S(\varphi_\ell)(b_{j,u}^*)
 =\sum_{i'=0}^{i}\sum_{v'=0}^{e-1}
 \left(\sum_{\alpha=0}^{e^2-1}
 \varphi_\ell(a_{i',\alpha})2^{\alpha(j+u+v')}\right)c_{j+i',v'}.
\]
Hence, the coefficient of $c_{j+i,v}$ is given by
\[
 \sum_{\alpha=0}^{e^2-1}\varphi_\ell(a_{i,\alpha})
       2^{\alpha(j+u+v)}
 =f_\ell(2^{j+u+v})=f_\ell(2^{u+v}2^j).
\]
Identifying $B_j$ and $C_{j+i}$ with two copies of $\CC^e$
by $b_{j,u}\mapsto b_u$ and $c_{j+i,v}\mapsto c_v$, the projection
of $S$ onto $E\otimes B_j\otimes C_{j+i}$ is thus $Q_f(2^j)$ (see \Cref{def:local-evaluation}).
Hence, the diagonal block of $K$ indexed by $j$ is
$\Kos{e-1}{Q_f(2^j)}=M_f(2^j)$.

Since a block upper triangular matrix has rank at least the sum
of the ranks of its diagonal blocks, \Cref{lem:local-nonzero} gives
\[
 \begin{aligned}
 \rank\Kos{e-1}{S}
 &\ge\rank K\ge\sum_{j=0}^{me-i-1}\rank M_f(2^j)\\
 &=e\binom{s}{e-1}(me-i)-\sum_{j=0}^{me-i-1}\dim\ker M_f(2^j)\\
 &\ge e\binom{s}{e-1}(me-i-e^2+1).
 \end{aligned}
\]
Restriction cannot increase border rank, so
\Cref{prop:koszul-bound} yields
\[
 \BR(T_{e,m}|_W)
 \ge\BR(S) \geq \frac{\rank \Kos{e-1}{S}}{\binom{s-1}{e-1}}
 =\frac{s}{e\binom{s}{e-1}}\rank\Kos{e-1}{S}
 \ge s(me-m-e^2),
\]
using $i<m$.
\end{proof}

\begin{theorem}[\Cref{thm:explicit} in more detail]\label{thm:uniform-bound}
For $n\ge16$, set
\begin{equation}\label{eq:uniform-parameters}
 e=\lfloor n^{1/4}\rfloor,
 \qquad
 m=\left\lfloor\frac{n}{e^2}\right\rfloor,
\end{equation}
and obtain $T_n\in(\CC^n)^{\otimes3}$ by zero-padding
$T_{e,m}$. Then
\[
 \BR(T_n)\ge3n-11n^{3/4}.
\]
\end{theorem}

\begin{proof}
Since $n\ge16$, we have $e\ge2$ and $n\ge e^4$, hence
$m\ge e^2$. Moreover, $s=2e-1\le e^2$, so $q=ms\le me^2=\dim A$. Recall also that $T_{e,m}$ is $A$-concise. Hence, we can apply \Cref{lem:substitution,lem:uniform-restrictions}
with $q=ms$ to obtain
\begin{equation}\label{eq:finite-simple}
 \begin{aligned}
 \BR(T_{e,m})
 &\ge me^2-ms+s(me-m-e^2)\\
 &=3me^2-5me-2e^3+2m+e^2\\
 &\ge3me^2-5me-2e^3.
 \end{aligned}
\end{equation}
Zero-padding preserves border rank, so this inequality, together
with $me^2\ge n-e^2$ and $me\le n/e$, gives
\begin{equation}\label{eq:uniform-error}
 \BR(T_n)\ge3n-\frac{5n}{e}-2e^3-3e^2.
\end{equation}
Since $n^{1/4}<e+1\le3e/2$,
\[
 \frac{5n}{e}\le\frac{15}{2}n^{3/4},
 \qquad
 2e^3+3e^2\le\frac72e^3\le\frac72n^{3/4},
\]
proving the bound.
\end{proof}

\section{Proof of the flattening barrier}
\label{sec:barrier-tensor}

In this section we prove the upper bound $\Pi(n) \leq 2n-2+\frac{2}{n+1}$ appearing in~\Cref{thm:barrier}. As described in~\Cref{sec:overview}, we need to prove that the matrices defined in~\Cref{eq:HK} are equal to the identity matrix, and to establish~\Cref{eq:moment_bounds}. We prove these statements in two lemmas.

\begin{lemma}\label{lem:tensor-alt-balance}
Let $L:(\CC^n)^{\otimes3}\to\Mat_{u,v}(\CC)$ be nonzero, and set
$r=\max_T\rank L(T)$ and $\rho=\rho(L)$.
There are matrices $M_{ijk}\in\Mat_r(\CC)$ satisfying
\begin{equation}\label{eq:tensor-alt-balance}
 \sum_{ijk}M_{ijk}M_{ijk}^{\dagger}
 =\sum_{ijk}M_{ijk}^{\dagger}M_{ijk}=I_r
\end{equation}
and
\begin{equation}\label{eq:tensor-alt-rank}
 \rank\left(\sum_{ijk}x_i y_j z_kM_{ijk}\right)\le\rho
 \qquad\text{for all }x,y,z\in\CC^n.
\end{equation}
\end{lemma}

We provide a self-contained, elementary proof of this lemma, although we note that it can also be proven as a consequence of the Kempf--Ness theorem for tensors presented in \cite[Theorem 4]{BGGOW2018}. To see this, project $L$ to a nonzero map $S:(\CC^n)^{\otimes3} \rightarrow \Mat_r(\CC)$ for which $S(T_*)$ has full rank for some $T_*$.  Viewing this map as a third-order $n^3 \times r \times r$ tensor, the fact that $S(T_*)$ has full rank implies that $S$ does not lie in the null cone of $\mathrm{SL}_r \times \mathrm{SL}_r$.  Hence, by \cite[Theorem 4]{BGGOW2018} there exists a tensor in its orbit closure with the properties given in~\Cref{eq:tensor-alt-balance} (up to rescaling). The rank constraint in~\Cref{eq:tensor-alt-rank} survives because matrix rank is invariant under $\mathrm{SL}_r \times \mathrm{SL}_r$, and the rank constraint is a closed condition.

\begin{proof}
Choose $T_*$ with $\rank L(T_*)=r$.
By post-composing $L$ with the row and column projections onto the
coordinates of an invertible $r\times r$ submatrix of $L(T_*)$,
we obtain matrices $X_{ijk}\in\Mat_r(\CC)$ with
\[
 \det\left(\sum_{ijk}(T_*)_{ijk}X_{ijk}\right)\ne0,
 \qquad
 \rank\left(\sum_{ijk}x_i y_j z_kX_{ijk}\right)\le\rho
\]
for all $x,y,z\in\CC^n$.

Let $\mathcal C$ be the set of all tuples
$(X_{ijk})\in\Mat_r(\CC)^{n^3}$ for which
\[
 \sum_{ijk}\|X_{ijk}\|_2^2\le r,
 \qquad
 \rank\left(\sum_{ijk}x_i y_j z_kX_{ijk}\right)\le\rho
 \quad\text{for all }x,y,z\in\CC^n,
\]
and note that $\mathcal C$ is compact. Choose a tuple $(M_{ijk})$
in $\mathcal C$ maximizing
\[
 \delta=\left|\det\left(\sum_{ijk}(T_*)_{ijk}M_{ijk}\right)\right|.
\]
The maximum is positive, since a sufficiently small nonzero
rescaling of the original tuple belongs to $\mathcal C$
and has a nonzero determinant.

Let $H=\sum_{ijk}M_{ijk}M_{ijk}^{\dagger}$ and note that $\tr(H)\le r$.
This matrix is positive
definite: a vector in its kernel is annihilated by every
$M_{ijk}^{\dagger}$ and therefore by the adjoint of the invertible matrix
$\sum_{ijk}(T_*)_{ijk}M_{ijk}$.
The tuple $\widetilde M_{ijk}=H^{-1/2}M_{ijk}$ also belongs to
$\mathcal C$, because left multiplication preserves the rank
constraints and
\[
 \sum_{ijk}\|\widetilde M_{ijk}\|_2^2
 =\tr(H^{-1/2}HH^{-1/2})=r.
\]
Note also that
\[
 \left|\det\left(\sum_{ijk}(T_*)_{ijk}\widetilde M_{ijk}\right)\right|
 =\frac{\delta}{(\det H)^{1/2}}.
\]
Maximality and $\delta>0$ imply $\det H\ge1$.
On the other hand, the arithmetic--geometric mean inequality for
the eigenvalues of $H$ gives
\[
 1\le\det H
 \le\left(\frac{\tr H}{r}\right)^r
 \le1.
\]
Equality forces every eigenvalue of $H$ to be $1$, so $H=I_r$.

Similarly, $K=\sum_{ijk}M_{ijk}^{\dagger}M_{ijk}$ is positive definite, and a similar argument with the matrices
$(M_{ijk}K^{-1/2})_{ijk}$ gives $K=I_r$.
\end{proof}

\begin{lemma}\label{lem:tensor-alt-moment-bound}
Let $n\ge2$, let $M_{ijk}\in\Mat_r(\CC)$, and let
\[
 H=\sum_{ijk}M_{ijk}M_{ijk}^{\dagger},\qquad
 K=\sum_{ijk}M_{ijk}^{\dagger}M_{ijk}.
\]
Choose $x,y,z$ independently and uniformly on the unit sphere in
$\CC^n$, and set
$M=\sum_{ijk}x_i y_j z_kM_{ijk}$.
Then
\begin{align*}
 \mathbb E\|M\|_2^2=\frac{\tr H}{n^3}, \qquad \mathbb E\|M\|_4^4\le\frac{\|H\|_2^2+\|K\|_2^2}{n^4(n+1)}.
\end{align*}
\end{lemma}

We note that the equality is well-known. By contrast, we are not aware of this upper bound on $\mathbb E\|M\|_4^4$ elsewhere in the literature.

\begin{proof}
For $x$ uniform on the unit sphere in $\CC^n$, it holds that
\begin{equation}\label{eq:tensor-alt-moments}
 \mathbb E[xx^{\dagger}]=\frac{I_n}{n},\qquad
 \mathbb E[xx^{\dagger}\otimes xx^{\dagger}]
 =\frac{I_n\otimes I_n+F}{n(n+1)},
\end{equation}
where $F(u\otimes v)=v\otimes u$ interchanges the two factors;
see \cite[Proposition~6]{Harrow2013}.
Taking matrix entries, using
$(I_n\otimes I_n)_{(a,c),(b,d)}=\delta_{ab}\delta_{cd}$ and
$F_{(a,c),(b,d)}=\delta_{ad}\delta_{cb}$, gives
\[
 \mathbb E[x_a\overline{x_b}]=\frac{\delta_{ab}}n,\qquad
 \mathbb E[x_a\overline{x_b}x_c\overline{x_d}]
 =\frac{\delta_{ab}\delta_{cd}+\delta_{ad}\delta_{cb}}{n(n+1)}.
\]
The first identity and independence give
\[
 \begin{aligned}
 \mathbb E\|M\|_2^2
 &=\sum_{\substack{i,j,k\\i',j',k'}}
   \mathbb E[x_i\overline{x_{i'}}]
   \mathbb E[y_j\overline{y_{j'}}]
   \mathbb E[z_k\overline{z_{k'}}]
   \tr\bigl(M_{ijk}M_{i'j'k'}^{\dagger}\bigr)\\
 &=\frac1{n^3}\sum_{ijk}\|M_{ijk}\|_2^2
 =\frac{\tr H}{n^3}.
 \end{aligned}
\]
This proves the equality.

To prove the inequality, encode the coefficient matrices
as a tensor
\[
 V\in P\otimes Q\otimes A\otimes B\otimes C,
 \qquad
 V_{\alpha,\beta,i,j,k}=(M_{ijk})_{\alpha\beta},
\]
where $P,Q$ are copies of $\CC^r$ and $A,B,C$ are copies of
$\CC^n$, equipped with their standard orthonormal bases. We use coordinates to identify all spaces with their duals.
Put $w=V\otimes\overline V$, with conjugation taken entrywise.
For each factor $D$, define
\[
 \Omega_D=\sum_{i=1}^{\dim D}e_i^D\otimes e_i^D,
 \qquad
 J_D=\Omega_D\Omega_D^\dagger,
\]
where $e_i^D$ is a standard basis of $D$. The operator $J_D$ acts on the two copies of $D$ and as the
identity elsewhere. We suppress these identity matrices to ease the notation. We also suppress implicit permutations of the five factors.
Write $\langle X\rangle:=w^\dagger Xw$ and
$J_S=\prod_{D\in S}J_D$ for a subset $S$ of the five factors.

Let $V_S$ be the tensor $V$ viewed as a map from the factors in $S^c$ to the factors in $S$. Then
\begin{equation}\label{eq:tensor-positive-complements}
 \langle J_S\rangle
 =\|V_S^\dagger V_S\|_2^2
 =\|V_SV_S^\dagger\|_2^2
 =\langle J_{S^c}\rangle.
\end{equation}
In particular, note that $V_P V_P^{\dagger}=H$ and $V_QV_Q^\dagger=K^{\mathsf T}$. Hence,
\[
 \begin{aligned}
 \langle J_P\rangle
 &=\|V_PV_P^\dagger\|_2^2=\|H\|_2^2,\\
 \langle J_Q\rangle&=\|K^{\mathsf T}\|_2^2=\|K\|_2^2.
 %
 \end{aligned}
\]
Set
\[
 \eta=\langle J_P+J_Q\rangle=\|H\|_2^2+\|K\|_2^2,
 \qquad
 \Lambda=\left\langle
 J_P(I+J_A)(I+J_B)(I+J_C)
 \right\rangle.
\]

Let
\[
\phi:=\sum_{\alpha,\beta} M_{\alpha,\beta}(x,y,z)  \; e_{\alpha}^P \otimes e_{\beta}^Q \in P\otimes Q
\]
be the vectorization of $M$. We have
\[
 \begin{aligned}
 (\phi\otimes\overline\phi)^\dagger
 J_P(\phi\otimes\overline\phi) =\|M^\dagger M\|_2^2=\|M\|_4^4.
 \end{aligned}
\]
Note that we also have
\[
 \bigl[
 I_{(P\otimes Q)^{\otimes2}}
 \otimes(\overline x\otimes x)^\dagger
 \otimes(\overline y\otimes y)^\dagger
 \otimes(\overline z\otimes z)^\dagger
 \bigr]w
 =\phi\otimes\overline\phi.
\]

For $D\in\{A,B,C\}$ and $u\in D$, put
$R_D(u):=(\overline u\otimes u)(\overline u\otimes u)^\dagger$,
acting on the two copies of $D$.
\Cref{eq:tensor-alt-moments},
with the first factor transposed, gives
\[
 \mathbb E\bigl[
 R_D(u)
 \bigr]
 =\frac{I+J_D}{n(n+1)}.
\]
Hence,
\begin{equation}\label{eq:tensor-positive-moment}
 \begin{aligned}
 \Lambda
 &=[n(n+1)]^3\left\langle
   J_P\,\mathbb E[R_A(x)]\mathbb E[R_B(y)]\mathbb E[R_C(z)]
   \right\rangle\\
 &=[n(n+1)]^3\,\mathbb E\left[
   (\phi\otimes\overline\phi)^\dagger
   J_P(\phi\otimes\overline\phi)\right]\\
 &=[n(n+1)]^3\,\mathbb E\|M\|_4^4.
 \end{aligned}
\end{equation}

The complements of $\{P,A,B,C\}$ and $\{Q,A,B,C\}$ among the five
factors are $\{Q\}$ and $\{P\}$, respectively. Thus
\Cref{eq:tensor-positive-complements} gives
\[
 \begin{aligned}
 \langle(J_P+J_Q)J_AJ_BJ_C\rangle
 &=\langle J_PJ_AJ_BJ_C\rangle+\langle J_QJ_AJ_BJ_C\rangle\\
 &=\langle J_Q\rangle+\langle J_P\rangle=\eta.
 \end{aligned}
\]

By \Cref{eq:tensor-positive-complements}, for every
$S\subseteq\{A,B,C\}$,
\[
 \langle J_QJ_S\rangle
 =\langle J_PJ_{\{A,B,C\}\setminus S}\rangle.
\]
It follows that
\[
 \begin{aligned}
 \sum_{D\in\{A,B,C\}}\langle(J_P+J_Q)J_D\rangle
 &=\sum_{\substack{S\subseteq\{A,B,C\}\\|S|=2}}
   \langle(J_P+J_Q)J_S\rangle=\Lambda-\eta.
 \end{aligned}
\]
For $D\in\{A,B,C\}$, we have
$J_D=\Omega_D\Omega_D^\dagger$ and $\|\Omega_D\|_2^2=n$,
so $nI-J_D\succeq0$.
These operators and $J_P+J_Q$ are positive semidefinite and
commute, so their product is positive semidefinite.
Expanding its expectation and using the preceding identities gives
\[
 \begin{aligned}
 0
 &\le\left\langle
 (J_P+J_Q)(nI-J_A)(nI-J_B)(nI-J_C)
 \right\rangle\\
 &=(n^3-1)\eta-(n^2-n)(\Lambda-\eta)\\
 &=(n-1)\bigl((n+1)^2\eta-n\Lambda\bigr).
 \end{aligned}
\]
Since $n\ge2$, we obtain $\Lambda\le(n+1)^2\eta/n$.
By \Cref{eq:tensor-positive-moment},
\[
 \mathbb E\|M\|_4^4
 =\frac{\Lambda}{[n(n+1)]^3}
 \le\frac{\eta}{n^4(n+1)}
 =\frac{\|H\|_2^2+\|K\|_2^2}{n^4(n+1)},
\]
as required.
\end{proof}

Putting these ingredients together, the proof is almost immediate.

\begin{proof}[Proof of the upper bound in \Cref{thm:barrier}]
Fix a nonzero $L$, and set $r=\max_T\rank L(T)$ and $\rho=\rho(L)$.
Choose the matrices supplied by \Cref{lem:tensor-alt-balance},
and let $M=\sum_{ijk}x_i y_j z_kM_{ijk}$ for independent uniform
$x,y,z$ on the unit sphere in $\CC^n$.
By \Cref{eq:tensor-alt-balance,lem:tensor-alt-moment-bound},
\[
 \mathbb E\|M\|_2^2=\frac r{n^3},
 \qquad
 \mathbb E \|M\|_4^4 \le\frac{2r}{n^4(n+1)}.
\]
For every $x,y,z$, we have $\rank M\le\rho$.
Cauchy--Schwarz for the nonzero squared singular values of $M$
gives $\|M\|_2^4\le\rho\, \|M\|_4^4$.
Another application of Cauchy--Schwarz therefore yields
\[
 \frac{r^2}{n^6}=(\mathbb E\|M\|_2^2)^2
 \le\mathbb E\|M\|_2^4
 \le\rho\,\mathbb E \|M\|_4^4
 \le\frac{2\rho r}{n^4(n+1)}.
\]
Multiplying by $n^6/(\rho r)>0$ gives $\Pot(L)\le2n^2/(n+1) = 2n-2 + \frac{2}{n+1}$.
Taking the supremum over $L$ proves the upper bound.
\end{proof}

\Needspace{12\baselineskip}
\section{Sharpness of the flattening barrier in odd dimensions}\label{sec:sharpness}

In this section, we prove that the upper bound $\Pi(n)\leq 2n-2+\frac{2}{n+1}$ appearing in~\Cref{thm:barrier} is sharp for all odd $n=2p+1\ge3$, and that this bound is attained by a Koszul flattening. To do this, it suffices to construct a tensor $S_n \in (\CC^n)^{\otimes 3}$ for which $\Kos{p}{S_n}$ is invertible. Indeed, from here it follows from~\Cref{prop:koszul-bound} that
\[
\Pot(\mathcal{K}_p)=\frac{(2p+1) \binom{2p+1}{p}}{\binom{2p}{p}}=2n-2+\frac{2}{n+1}.
\]
The tensor $S_n$ is constructed as follows.
\begin{definition}
\label{def:odd-fifth-tensor}
Let $A=\CC^n$ have basis $a_1,\ldots,a_n$, identify
$B^*=\Poly{n}$, and put $C=\CC^{p+1}\oplus\CC^p$.
For $v_q(t)=(1,t,\ldots,t^{q-1})^{\mathsf T}$, define
\begin{equation}\label{eq:odd-fifth-tensor}
 S_n=\sum_{i=1}^n a_i\otimes X_i,
 \qquad
 X_i(f)=f(i)v_{p+1}(i)\oplus f'(i)v_p(i).
\end{equation}
\end{definition}

\begin{theorem}
\label{thm:odd-fifth-invertibility}
For the tensor $S_n$ introduced in \Cref{def:odd-fifth-tensor},
$\Kos{p}{S_n}$ is invertible.
\end{theorem}

To prove this theorem, we require two lemmas. The first is a straightforward application of Lagrange
interpolation. For example, it can be proven using Euler's identity
(see, e.g., \cite[Equation~(2)]{Nica2023}). We give a self-contained proof for completeness.

For $S\subseteq[n]$, let
\[
 \Delta_S=\prod_{\substack{i<j\\i,j\in S}}(j-i),
\]
with empty products equal to $1$.

\begin{lemma}
\label{lem:odd-fifth-derivative-compatibility}
If $f\in\Poly{n}$ vanishes on a $(p+1)$-subset $S\subseteq[n]$, then
\[
 \sum_{i\in S}\Delta_{S\setminus\{i\}}^2 f'(i)=0.
\]
\end{lemma}

\begin{proof}
Put $h_i(t)=(t-i)\prod_{j\in S\setminus\{i\}}(t-j)^2$.
Define
\[
 q(t)=\Delta_S^2 f(t)
 -\sum_{i\in S}\Delta_{S\setminus\{i\}}^2 f'(i)h_i(t).
\]
For each $k\in S$, we have $f(k)=0$ and $h_i(k)=0$ for every
$i\in S$, so $q(k)=0$. If $i\ne k$, the factor $(t-k)^2$ in
$h_i(t)$ also gives $h_i'(k)=0$, while
\[
 h_k'(k)=\prod_{j\in S\setminus\{k\}}(k-j)^2
 =\frac{\Delta_S^2}{\Delta_{S\setminus\{k\}}^2}.
\]
Consequently,
\[
 q'(k)=\Delta_S^2 f'(k)
 -\Delta_{S\setminus\{k\}}^2 f'(k)h_k'(k)=0.
\]
Thus $q$ has a double root at each of the $p+1$ points of $S$.
Since its degree is at most $2p+1$, it is zero.
Each $h_i$ is monic of degree
$2p+1$, whereas $\deg f\le2p$, so comparison of leading
coefficients proves the claim.
\end{proof}

\begin{lemma}\label{lem:odd-fifth-subset-sums}
If scalars $c_J$, indexed by the $(p+1)$-subsets of $[n]$, satisfy
\[
 \sum_{i\notin I}c_{I\cup\{i\}}=0
 \qquad\text{for all $p$-subsets }I\subseteq[n],
\]
then every $c_J=0$.
\end{lemma}

\begin{proof}
Fix $S\subseteq[n]$ with $|S|\le p$. Summing the equations over
all $p$-subsets $I$ containing $S$ gives
\[
 0=\sum_{\substack{|I|=p\\S\subseteq I}}
       \sum_{i\notin I}c_{I\cup\{i\}}
  =(p+1-|S|)
   \sum_{\substack{|J|=p+1\\S\subseteq J}}c_J.
\]
Fix a $(p+1)$-subset $K$. For every $S\subseteq K^c$, we have
$|S|\le p$, so the preceding identity gives
$\sum_{\substack{|J|=p+1\\S\subseteq J}}c_J=0$.
Hence,
\begin{align*}
 0&=\sum_{S\subseteq K^c}(-1)^{|S|}
           \sum_{\substack{|J|=p+1\\S\subseteq J}}c_J\\
  &=\sum_{|J|=p+1}c_J
           \sum_{S\subseteq J\cap K^c}(-1)^{|S|}.
\end{align*}
Grouping the subsets in the inner sum by size and applying the
binomial theorem gives
\[
 \sum_{S\subseteq J\cap K^c}(-1)^{|S|}
 =\sum_{s=0}^{|J\cap K^c|}(-1)^s\binom{|J\cap K^c|}{s}
 =\begin{cases}
   1,&J\cap K^c=\varnothing,\\
   0,&J\cap K^c\ne\varnothing.
  \end{cases}
\]
Since $|J|=|K|=p+1$, the condition $J\cap K^c=\varnothing$
is equivalent to $J=K$. Thus only the term $c_K$ remains in the
sum, proving $c_K=0$.
\end{proof}

\begin{proof}[Proof of ~\Cref{thm:odd-fifth-invertibility}]
Let $w$ be a vector in the kernel of $\Kos{p}{S_n}$. We can write
\[
 w=\sum_{|I|=p}a_I\otimes f_I,\qquad f_I\in\Poly{n}.
\]
For each $(p+1)$-subset $J$, write $J=\{j_0<\cdots<j_p\}$.
By \Cref{eq:koszul-definition},
\begin{align*}
 0=\Kos{p}{S_n}w
  &=\sum_{|I|=p}\sum_{i\notin I}
     (a_i\wedge a_I)\otimes X_i(f_I)\\
  &=\sum_{|J|=p+1}a_J\otimes
     \left(\sum_{r=0}^p(-1)^rX_{j_r}(f_{J\setminus\{j_r\}})\right).
\end{align*}
Thus,
for every $(p+1)$-subset $J$,
\[
 \sum_{r=0}^p(-1)^rX_{j_r}(f_{J\setminus\{j_r\}})=0.
\]
By the definition of $X_i$, the two blocks of this equation are
\begin{align*}
 \sum_{r=0}^p(-1)^r
 f_{J\setminus\{j_r\}}(j_r)v_{p+1}(j_r)&=0,\\
 \sum_{r=0}^p(-1)^r
 f'_{J\setminus\{j_r\}}(j_r)v_p(j_r)&=0.
\end{align*}
Since $v_{p+1}(j_r)=(1,j_r,\ldots,j_r^p)^{\mathsf T}$,
the first vector equation is equivalent to the $p+1$ scalar equations
\[
 \sum_{r=0}^p j_r^k
 \left((-1)^r f_{J\setminus\{j_r\}}(j_r)\right)=0
 \qquad(0\le k\le p).
\]
Since the square matrix $[j_r^k]_{0\le k,r\le p}$ is invertible,
it follows that $f_I(i)=0$ for all $p$-subsets $I\subseteq[n]$
and all $i\notin I$.

The second vector equation is equivalent to the $p$ scalar equations
\[
 \sum_{r=0}^p j_r^k
 \left((-1)^r f'_{J\setminus\{j_r\}}(j_r)\right)=0
 \qquad(0\le k<p).
\]
This system has coefficient matrix
$[j_r^k]_{0\le k<p,\,0\le r\le p}$. Deleting column $r$ leaves
a square Vandermonde matrix with determinant
$\Delta_{J\setminus\{j_r\}}\ne0$, so this matrix has rank $p$
and a one-dimensional kernel. For $0\le k<p$, adding
$(j_0^k,\ldots,j_p^k)$ as a new first row gives a square matrix
with two equal rows. Expanding its zero determinant along the
first row gives
\[
 \sum_{r=0}^p(-1)^r j_r^k\Delta_{J\setminus\{j_r\}}=0
 \qquad(0\le k<p).
\]
Thus the nonzero vector
$\bigl((-1)^r\Delta_{J\setminus\{j_r\}}\bigr)_{r=0}^p$
lies in the kernel and spans it. The vector
$\bigl((-1)^r f'_{J\setminus\{j_r\}}(j_r)\bigr)_{r=0}^p$
is therefore a scalar multiple of this vector. Hence, there is a scalar $\alpha_J$ such that
\begin{equation}\label{eq:odd-fifth-kernel-data}
 f_{J\setminus \{i\}}(i)=0,\qquad
 f'_{J \setminus \{i\}}(i)=\alpha_{J}\Delta_{J \setminus \{i\}}
 \qquad\text{for every $i \in J$}.
\end{equation}

Fix a $p$-subset $I \subseteq [n]$, and recall that $f_I$ vanishes on $I^c$. Applying \Cref{lem:odd-fifth-derivative-compatibility}, and \Cref{eq:odd-fifth-kernel-data} for each $i \notin I$, gives
\[
 0=\sum_{i\notin I}\Delta_{(I\cup\{i\})^c}^2 f'_I(i)
  =\Delta_I\sum_{i\notin I}
       \alpha_{I\cup\{i\}}\Delta_{(I\cup\{i\})^c}^2.
\]
Since $\Delta_I\ne0$, \Cref{lem:odd-fifth-subset-sums} applied to
$c_J=\alpha_J\Delta_{J^c}^2$ gives $\alpha_J=0$ for every $J$.
By \Cref{eq:odd-fifth-kernel-data}, $f_I$ has $p+1$ distinct double
roots but degree at most $2p$, so $f_I=0$.
Hence $\Kos{p}{S_n}$ is injective. It is square because
$\binom np=\binom n{p+1}$, and therefore invertible.
\end{proof}

%

%

\begingroup
\small
\setlength{\bibsep}{4pt}
\interlinepenalty=10000
\bibliographystyle{unsrtnat}
\bibliography{references}
\endgroup

\end{document}